\documentclass[journal,10pt]{IEEEtran}
\usepackage{amsmath,amssymb,amsthm,bm}
\usepackage{booktabs,array,tabularx,graphicx}
\usepackage{algorithm,algpseudocode}
\usepackage{tikz}
\usetikzlibrary{arrows.meta,positioning}
\usepackage{cite}
\usepackage{xurl}
\usepackage[hidelinks,unicode]{hyperref}
\graphicspath{{figures/}}
\newcommand{\us}{\,\mu\mathrm{s}}
\newcommand{\ms}{\,\mathrm{ms}}
\newcommand{\E}{\mathbb{E}}
\newcommand{\Prob}{\mathbb{P}}
\newcommand{\range}{\operatorname{range}}

\newtheorem{proposition}{Proposition}
\newcolumntype{Y}{>{\raggedright\arraybackslash}X}
\begin{document}
\title{Timing and Reliability Limits of Position-Assisted UWB Multicast Audio}
\author{Fengyun~Zhang,~\IEEEmembership{Member,~IEEE},~Member,~AES%
\thanks{Fengyun Zhang is with the College of Artificial Intelligence, Southwest University, Chongqing, China (e-mail: zhangfy2019@mail.sustech.edu.cn).}}
\markboth{IEEE Transactions on Multimedia --- Manuscript}%
{Zhang: Timing and Reliability Limits of Position-Assisted UWB Multicast Audio}
\maketitle
\begin{abstract}
One-to-many ultra-wideband (UWB) transport can disseminate shared audio blocks without replicating a downlink stream for every receiver. Yet common reception does not ensure aligned electrical playback, coincident acoustic arrival, or delivery to an entire receiver group. This paper develops a unified feasibility analysis of position-assisted UWB multicast audio that separates clock and actuation error, age-dependent acoustic propagation compensation, correlated block erasures, and reserved communication resources. We derive a conditional bound on one-listener arrival spread, a stationary expression for all-receiver block delivery under common and receiver-local losses, and complementary bounds for reliability-constrained admission. Reproducible Monte Carlo experiments examine 4,590 policy--run records across clock, localization, channel, and slot configurations. For eight receivers, affine clock correction produces a mean run-level 95th-percentile electronic span of $21.19\us$, while position assistance reduces the corresponding one-listener acoustic arrival span from $21.95\ms$ to $419.43\us$. Under the stated $30\us$ timing criterion, 98.86\% of compensated events still exceed the threshold. For 64 receivers, three block copies attain an all-receiver delivery probability of approximately 0.97478 under the modeled correlated channel. The analysis exposes the distinct failure mechanisms of position age, shared fades, retransmission limits, and schedule saturation, and establishes joint criteria for interpreting synchronous multicast audio performance.
\end{abstract}
\begin{IEEEkeywords}
ultra-wideband, multicast audio, acoustic synchronization, broadcast reliability, correlated erasures.
\end{IEEEkeywords}
\section{Introduction}\label{sec:intro}
\IEEEPARstart{D}{istributed} audio systems increasingly deliver common source material to multiple independent playback devices. Paired listening devices, multi-speaker installations, localized interpretation systems, and assistive listening services differ in their presentation requirements, yet share a central engineering difficulty: identical content must reach the appropriate output interface at the appropriate time. A common received packet does not specify the corresponding sample index, fix receiver clock drift, determine the digital-to-analog conversion instant, or eliminate the difference in sound-propagation paths. For spatially distributed loudspeakers, a system may therefore achieve high packet delivery while producing substantial inter-device temporal mismatch or unequal acoustic arrival at the intended listener. A rigorous service model must distinguish these quantities before evaluating a wireless transport architecture.

Standardized Bluetooth Low Energy (LE) Audio provides an important reference for this problem. Connected Isochronous Streams support addressed links, whereas Broadcast Isochronous Streams within a Broadcast Isochronous Group provide common one-to-many transport; Auracast builds an interoperable broadcast-audio service on that mechanism~\cite{bhalla,bison,auracast}. Broadcast distribution itself is consequently not new, nor is reduced duplication a unique property of ultra-wideband (UWB). UWB is of interest because it offers an alternative short-range physical and scheduling platform for applications that may combine fine-grained timing control, spatial information, and local broadcast dissemination. Whether such capabilities translate into acoustic synchronization and reliable group playback depends on the receiver implementation and the constraints imposed by the content deadline. This distinction prevents a radio-level capability from being mistaken for an end-to-end audio guarantee.

Three interacting challenges make the multicast-audio problem more stringent than conventional packet delivery. First, receivers own distinct oscillator, buffering, sample-rate, and actuation domains. Clock-offset compensation alone cannot guarantee simultaneous electrical output when timestamp observations contain bias or the audio samples are committed at different device-specific instants. Second, a geographically distributed speaker group has unequal distances to a selected listener. Compensating the resulting acoustic path delays requires estimates whose error depends both on coordinate uncertainty and on the age of the most recent position update. Improving clock synchronization without accounting for this geometric term can yield a small electronic spread while leaving a much larger acoustic spread. Third, a broadcast block is useful for group rendering only when every designated receiver obtains it before the applicable deadline. Receiver-local erasures, burst-correlated common losses, retransmission opportunities, and feedback scheduling jointly determine this event. Average per-receiver delivery probability and all-receiver delivery probability are not interchangeable.

These challenges are also coupled through limited time and communication resources. A repeated broadcast can increase the chance of block delivery, but every scheduled copy occupies service capacity and potentially reduces the time available for decoding and output scheduling. Receiver-specific acknowledgments or channel reports add a return-path reservation whose cost increases with group size. Position updates can decrease path-compensation error while consuming acquisition and processing time; an old but precise coordinate can be less useful than a recent, noisier observation during motion. Consequently, performance cannot be established by optimizing a clock residual, packet-delivery ratio, or downlink airtime in isolation. The operative question is whether the relevant timing, reliability, and capacity constraints have a nonempty common feasible region.

The problem considered in this paper comprises one UWB audio source, $N$ independently clocked receivers, application-level copies of each common audio block, and an external source of position information for the speaker-to-listener geometry. Each source block has a presentation time and is subject to a playback deadline. The output requirement is expressed using the spread of electrical emission times or, for spatial playback, the spread of direct-path arrivals at one designated listening position. Reliability is defined by the joint event that all designated receivers recover the same block, rather than by the mean delivery ratio over devices. A reserved-slot utilization constraint accounts for downlink copies, periodic timing service, and any receiver-specific feedback. These definitions establish comparable quantities for a feasibility analysis while making clear that alignment at one listener is not equivalent to coherent reproduction over an entire room.

Prior investigations support individual elements of this architecture but do not resolve their interaction under the same service objective. UWB musical-haptic communications have been demonstrated and evaluated on dedicated hardware by Turchet \emph{et al.}~\cite{turchet}, while Zhang has studied scheduling and clock correction in a UWB stereo-audio architecture~\cite{zhang}. Separate studies address UWB time synchronization, integrated communication and localization, and audio sample-rate correction~\cite{hollosi,liu,korse}. Their significance is specific to the observation models, radio interfaces, and output paths examined in those works. The present question is different: how timing errors, position-estimation age, state-correlated erasures, and service reservations constrain an entire multicast receiver population under an explicitly declared acoustic objective. This requires a common accounting framework rather than an assumed equivalence between radio timestamps, application callbacks, and loudspeaker output.

An additional methodological difficulty is the distinction between a documented wireless feature and the observable variables required by a mathematical system model. Addressed and broadcast connections can be supported by a software interface without exposing a calibrated antenna-edge timestamp or a sample-addressable digital-to-analog converter deadline. Similarly, the existence of an audio sample-rate adjustment mechanism does not imply common absolute playback phase across devices. We therefore ground the architecture in an audit of the available software and scheduling documentation, then formulate clock, position, and channel quantities as explicit model inputs. The analysis concerns a conditional operating region: every numerical statement is tied to its observation definition, distributional assumptions, and scheduling boundary, rather than being interpreted as a hardware specification.

A concrete timing criterion makes these distinctions visible. For a sinusoidal signal of frequency $f$, an output mismatch $\Delta t$ induces phase displacement $2\pi f\Delta t$. A $30\us$ separation corresponds to $10.8^\circ$ at 1~kHz and $108^\circ$ at 10~kHz. These values are direct phase calculations, not a psychoacoustic audibility threshold. We employ $30\us$ as an explicit numerical sensitivity criterion, allowing clock, actuator, and position errors to be related to a common target. Because the geometric correction is applied to a designated listening position, its direct-path benefit must also be interpreted spatially: a correction that equalizes arrivals at one point need not do so at other points. The value of the framework lies in making such design assumptions and tradeoffs mathematically inspectable.

\noindent\textbf{Contributions:} This paper makes three contributions. \emph{1) Joint timing and service formulation:} We establish a common mapping from source presentation times to receiver clock estimates, scheduled electrical emissions, and one-listener acoustic arrivals; a conditional error bound separates clock residual, output quantization, coordinate uncertainty, and information age, while a causal deadline inequality and reserved-load expression expose independent feasibility restrictions. \emph{2) Correlated group-delivery analysis:} For a two-state common-loss process with heterogeneous receiver-local erasures, we give the stationary all-receiver delivery expression, a necessary failure floor, and a sufficient union-bound test, thereby distinguishing individual reception from reliable simultaneous service and exposing when retransmission caps preclude a group target. \emph{3) Reproducible sensitivity evaluation:} We connect the documented broadcast and scheduling interfaces to common-seed comparisons of clock policies, position-error regimes, fixed and adaptive repetitions, receiver-count scaling, and slot overload; run-level statistics and explicit parameter provenance isolate the failure mechanisms relevant to multicast audio design.

The results show that low electronic timestamp dispersion can coexist with much greater position-dependent acoustic spread, while improved marginal reception can fail to deliver a common block across a large receiver group. In addition, increasing redundancy can move the system outside its reserved-slot capacity even when the underlying loss model favors more copies. Together, these observations turn familiar component-level improvements into falsifiable system-level feasibility statements: a multicast configuration is admissible only when its output-time, group-delivery, and resource requirements are simultaneously satisfied.

The remainder of the paper is organized as follows. Section~\ref{sec:related} reviews broadcast-audio transport and related timing and localization studies. Section~\ref{sec:platform} examines the network architecture and documented platform interfaces. Section~\ref{sec:model} introduces the system model and analytical bounds, and Section~\ref{sec:methods} presents the reference policies. Section~\ref{sec:evaluation} specifies the numerical protocol and statistical estimands; Section~\ref{sec:results} reports the results. Section~\ref{sec:discussion_main} discusses cross-layer implications and model scope, and Section~\ref{sec:conclusion} concludes the paper.

\section{Background and Related Work}\label{sec:related}
\subsection{Search scope and evidence hierarchy}
A targeted narrative literature review was completed on 17~October~2026 using Crossref, Google Scholar, arXiv metadata, and available publisher, author, institutional, and official standards pages. Queries combined UWB audio, LE Audio isochronous transmission, BIS reliability, sample-rate offsets, distributed acoustic synchronization, and integrated communication/localization. The strict recent window was 17~October~2021--17~October~2026; older references were retained only for established clock, filtering, burst-channel, and broadcast foundations. The 2021 LE Audio book is background outside that strict window, not evidence of a recent discovery~\cite{bhalla}.

The saved search log contains 151 candidate records, including nonrelevant matches, duplicates, and an invalid entry. Candidate count is not a screened-study count and does not imply completeness. Twenty-one scholarly works and two official document sources were selected for focused comparison. Selection favored directly relevant primary work, identifiable authors and dates, stable DOI or repository URLs, and access to a concrete abstract or passage supporting the claim. Because database coverage, access to full texts, and ranking are incomplete, this targeted review does not claim systematic-review coverage. The absence of a result cannot substantiate a first-ever or prior-art-exclusion claim.

Evidence is separated into peer-reviewed published work, explicitly labeled preprints, official specifications or vendor documentation, executed simulation, and architecture-level modeling assumptions. Published issue metadata is preferred over an earlier online date; for example, the localization survey by Hapsari \emph{et al.} appeared online in 2024 but has volume-13 bibliographic assignment in 2025~\cite{hapsari}. The UWB audio work of Zhang and the ultrasound architecture of Laurijssen \emph{et al.} remain labeled preprints~\cite{zhang,laurijssen}. Source titles alone support only narrow landscape descriptions, not quantitative performance claims.

\subsection{LE Audio, broadcast reliability, and group interaction}
LE Audio establishes an essential prior-art baseline. CISes provide acknowledged, connection-oriented isochronous transport, whereas BISes provide unacknowledged synchronized one-to-many delivery. The Bluetooth Core Specification explicitly permits unconditional BIS subevent retransmissions and advance transmission of packets~\cite{btcore,bison,mantis}. A broadcast isochronous group may contain multiple BISes; therefore, the phrase ``a single BIS for stereo'' is not a universal requirement. The number of streams, channel mapping, subevents, and presentation timing depend on the configured Broadcast Isochronous Group. A comparison against BIS must therefore account for its permitted retransmission and pretransmission mechanisms, not compare UWB copies against an artificially unrepeated broadcast. Auracast is an interoperable broadcast implementation and branding framework associated with the Public Broadcast Profile, not a guarantee that every LE Audio device implements broadcast reception~\cite{auracast}.

Yan \emph{et al.} analyze CIS loss, throughput, backlog, latency, and average power, including multihop configurations~\cite{yan}. Their work demonstrates that reliable protocol-level comparison requires scheduling, feedback, buffering, and PHY parameters rather than a count of logical links. It also differs from the present reduced-order study: we do not reproduce their stack, calibrate a BLE baseline, or equate raw PCM offered load with LC3 traffic. Earlier work by Lee examines broadcast audio over interfered BLE channels~\cite{lee}, demonstrating that the reliability problem predates Auracast branding and that broadcast is not intrinsically loss-free.

Recent group-communication work further narrows novelty claims. MANTIS introduces coordinated broadcaster/receiver role transitions to support multiparty LE Audio while retaining a single active broadcaster~\cite{mantis}. This is a different objective from independent one-way playback, but it already recognizes that the return path and coordination mechanism determine scalability. BISON demonstrates security risks in BIS synchronization and broadcast-code handling on off-the-shelf hardware~\cite{bison}. Its results motivate authentication and metadata protection; they are not treated as proof that every current Auracast implementation is vulnerable in the same way, nor as proof that UWB is automatically more secure.

\subsection{UWB audio, timing, and integrated localization}
The closest experimental precedent is the musical-haptics work of Turchet \emph{et al.}~\cite{turchet}. It evaluates SPARK SR1020 communication and complete haptic prototypes, including alignment of instrument sound and tactile delivery. The article explicitly warns that a latency advantage inferred from a reported Auracast value requires a comparison under the same experimental conditions. We adopt that evidentiary caution: its measured values are not re-labeled as measurements of our SR1100-based proposal, and cross-paper latency numbers are not used as a superiority benchmark.

Zhang's 2025 preprint discusses high-resolution lossless UWB audio and reports a two-board stereo setup~\cite{zhang}. Its architecture includes clock-drift compensation, audio pipelines, synchronization frames, and scheduling. These features overlap substantially with any generic claim to combine UWB with audio streaming. The present study differs in its explicit multi-receiver error/resource budgets and negative feasibility outcomes, not in claiming to invent those building blocks. The preprint's broad commercial and performance statements are not imported as independently verified facts.

UWB-based clock synchronization and integrated localization are also established research topics. Holl\'osi and Moldov\'an evaluate wireless UWB synchronization of IEEE~1588 clocks and discuss implementation-specific error sources~\cite{hollosi}. Liu \emph{et al.} develop asynchronous integrated communication and localization using IR-UWB signals, including clock-drift estimation and downlink localization~\cite{liu}. Such work establishes that joint communication and localization is possible with suitable waveform processing and hardware. It does not establish that a supplied SPARK public API exposes those functions or that a packet audio application inherits their accuracy. Surveys emphasize platform constraints and environment-dependent localization behavior~\cite{savic,hapsari}.

\subsection{Sampling clocks, acoustic alignment, and the remaining question}
Network-clock synchronization and sampling-clock synchronization solve related but different problems. Reference-broadcast synchronization removes some sender-side uncertainty by comparing common reception events~\cite{elson}; general clock surveys distinguish offset, skew, observation delays, and implementation effects~\cite{sundar}. Two-state Kalman filtering is a conventional estimator~\cite{kalman}; moving it to a different bearer is not, in isolation, an estimator contribution.

Acoustic sensor-network studies make the rate issue concrete. Gburrek \emph{et al.} consider time-varying sampling-rate offsets, changing speakers, and distance-assisted sampling-time alignment~\cite{gburrek22}; related work explicitly addresses integrating rate synchronization and beamforming~\cite{gburrek23}. Didier \emph{et al.} compensate coherence drift and full-sample shifts in distributed estimation~\cite{didier}. On the playback side, Korse \emph{et al.} estimate and compensate offsets for stereo rendering and evaluate binaural effects~\cite{korse}. These are stronger precedents for signal-domain synchronization than packet arrival coincidence.

Laurijssen \emph{et al.} describe a preprint architecture that deliberately separates optical position calibration, a sub-GHz synchronization signal, and Wi-Fi payload transport~\cite{laurijssen}. This is a useful counterexample to the premise that one radio must inherently supply all three functions. It concerns ultrasound acquisition, not a SPARK loudspeaker system, and its reported sub-microsecond result is not adopted as ours. Table~\ref{tab:related} summarizes the evidence boundary: the unresolved task is platform-specific joint feasibility, not the first conceptual combination of wireless transport, clock correction, and position information.

\subsection{Position of this study and unresolved technical coupling}
The closest available studies address important portions of the problem but use different observables and evaluation endpoints. The UWB audio prototype literature assesses an implemented audio path; the distributed-acoustic literature addresses sample-rate or output alignment; LE Audio defines an actual broadcast transport; and common-fade channel models characterize correlated reception. These lines of work do not by themselves establish the feasibility of a particular externally positioned, independently clocked UWB multicast installation. The unresolved coupling is operational: an application must have a timely, calibrated observation, a causal playback actuator, sufficiently reliable delivery to the entire intended group, and a legal schedule for all repeated blocks and control exchanges.

We consequently avoid two unsupported novelty shortcuts. First, a system is not new solely because it combines established clock tracking, propagation compensation, and repeated transmission. Second, improvement against a deliberately weaker single-copy or independently addressed baseline is not evidence of superiority to mature BIS operation. The appropriate claim is conditional: given a stated timestamp model, position-age distribution, erasure process, and service reservation, one can compute attainable and unattainable operating regions. The numerical study then tests which modeled mechanism dominates the violation of those conditions. Accordingly, the present comparison is restricted to the same parameterized service model and does not rank complete commercial implementations.

The standard/prototype distinction affects the choice of evidence. Bluetooth Core specifications establish permitted and mandatory behavior but do not supply measured behavior of arbitrary Auracast products. Vendor documentation establishes an interface contract but cannot substitute for calibrated board traces. Simulation can test a stochastic model without proving that the model describes a particular installation. These distinctions are applied consistently to the capability table, the numerical study, and the stated evaluation metrics.

\begin{table*}[t]
\caption{Related-work evidence and distinctions. Published results belong to the cited platforms; no numerical cross-platform ranking is inferred.}\label{tab:related}
\centering\small
\begin{tabularx}{\textwidth}{p{24mm}p{27mm}YYY}
\toprule
Work / period & Primary scope & Relevant established capability & Important difference & Evidence status \\
\midrule
Yan \emph{et al.}, 2024~\cite{yan} & CIS performance model & Loss, throughput, backlog, delay, power tradeoffs & Not a UWB multicast or calibrated DAC comparison & Published journal; analytical/simulation \\
BISON, 2023; MANTIS, 2026~\cite{bison,mantis} & BIS security / multiparty audio & Synchronized broadcast, return-path and coordination issues & Neither makes UWB broadcast or authentication automatic & Published conference; hardware demonstrations \\
Turchet \emph{et al.}, 2025~\cite{turchet} & SR1020 musical haptics & Real link/end-to-end tests, sound--touch timing & Different radio/application; no inherited SR1100 multicast result & Published journal; measurements and perception \\
Zhang, 2025~\cite{zhang} & Lossless UWB audio & Stereo architecture, clock-drift and schedule discussion & Two-board study; generic audio-over-UWB novelty already occupied & Preprint; reported setup not reproduced here \\
Gburrek / Didier / Korse, 2022--2025~\cite{gburrek22,didier,korse} & Acoustic sampling offsets & Rate/time distinction; signal-domain correction & Not equivalent to application RX timestamp alignment & Published conference/journal; methods and evaluations \\
Holl\'osi, 2023; Liu, 2026~\cite{hollosi,liu} & UWB clocks / integrated localization & Hardware-dependent timing and joint estimation & Different APIs and waveform assumptions & Published journal; platform-specific validation \\
This study & Multi-receiver feasibility & Explicit timing, common-loss, and schedule budgets & No real hardware data, codec, or commercial-stack benchmark & Executed assumed-model simulation + SDK audit \\
\bottomrule
\end{tabularx}
\end{table*}

\section{Architecture and Platform Capability Audit}\label{sec:platform}
\subsection{Topologies and what a comparison must hold constant}
We distinguish an inter-receiver decode-and-forward topology, independently addressed source-to-receiver streams, standardized BIS broadcast, and application-level UWB PAN broadcast. The last two share the key structural advantage of common downlink dissemination. Calling UWB a ``fourth topology'' is useful only as an implementation category, not a proof of a new graph-theoretic paradigm. Table~\ref{tab:topology} compares structural properties without attributing unverifiable private mechanisms to Apple or Qualcomm. Product-specific simultaneous listening, snooping, dual connections, and retransmission procedures require model/version evidence before being used as baselines.

The relevant comparison variables are the source samples, audio bandwidth and channel count, codec or raw-PCM representation, presentation deadline, receiver geometry, interference conditions, and the boundaries of latency and energy accounting. If compression differs, radio airtime and decoding delay must be reported separately. If one system uses an external position tracker or extra timing radio, its cost must be included. Our comparison instead isolates abstract repetition and clock policies under a common model. It does not emulate Apple, Qualcomm, LC3, CIS link-layer procedures, or Auracast discovery and presentation delay.

\begin{table*}[t]
\caption{Structural topology comparison, not measured platform performance. ``Per-stream'' costs assume non-overlapping service without spatial reuse.}\label{tab:topology}
\centering\small
\begin{tabularx}{\textwidth}{p{29mm}YYYY}
\toprule
Category & Common content service & Return-path issue & Timing issue & Limitation of current comparison \\
\midrule
Inter-receiver forwarding & At least two hop services for a two-receiver chain & Per-hop feedback and forwarding policy & Additional processing/buffering; not necessarily twice total latency & Structural lower-bound counter only \\
Independent addressed streams / CIS & Separate link service; synchronized groups possible~\cite{yan,bison} & Per-link control and scheduling & Group timing does not alone calibrate all DAC paths & No commercial or standards-complete stack simulated \\
BIS / Auracast & Shared broadcast; one or multiple BISes in a BIG~\cite{bison,auracast} & No ordinary receiver ACK in the BIS & Isochronous rendering; product-specific audio pipeline & No empirical latency/energy baseline \\
Audited UWB PAN broadcast & One connection to destination \texttt{0xFF}~\cite{spark} & ACK off; application copies and separate reports & Output timing modeled separately & Feasibility model; not a qualified interoperable audio profile \\
\bottomrule
\end{tabularx}
\end{table*}

\subsection{Capabilities established by the supplied SDK}
The audited input was an SR1100 SDK v2.4.0 archive and its accompanying documentation package. Archive extraction was checked for path integrity, expected sizes, and CRC consistency. The online index identifies version~2.4.0, released 24~July~2026; three retrieved source chapters---wireless concepts, schedule design, and audio CDC---match the corresponding attachment hashes byte-for-byte. This supports the provenance of the examined documentation, not the behavior of every firmware configuration or a particular physical board. The accompanying capability ledger records file paths, relevant lines, and SHA-256 digests for the reviewed materials~\cite{spark}; the proprietary SDK itself is not redistributed with this paper.

The wireless concepts describe destination \texttt{0xFF} as PAN broadcast, with addresses \texttt{0x00}--\texttt{0xFE} for nodes: 255 addressable active nodes in total, not a measured 255-receiver reliability guarantee. The API documentation warns that broadcast ACKs collide because all receivers would reply in the same timeslot. The current header states that retransmission does nothing with ACK disabled. Thus, the repetitions modeled below are separately queued \emph{application copies}; they are not unacknowledged invocations of working SWC automatic-repeat request.

Payload documentation is internally inconsistent: generic API prose gives a maximum range of 118--124~bytes, while the SR1120 schedule table states 255~bytes and a profiler example uses 250~bytes. We do not resolve this by guessing a product limit. The model uses 112~SWC bytes, conservatively below either stated range; header and feature allocations are configuration-dependent. The example audio backend includes SAI DMA callbacks and a MAX98091 codec path for a specific reference platform. This fact does not establish that an unspecified board provides a sample-addressable absolute playback deadline, an equivalent codec clock, or a functioning multicast firmware image.

\subsection{Capabilities not established and the architecture boundary}
A legacy 2023 changelog mentions ranging callbacks and phase-data settings. The reviewed current \texttt{core/wireless/swc\_api.h} does not expose those entries or an identified calibrated antenna-edge timestamp interface. This is evidence of an unestablished public-API capability, not proof that the silicon is physically incapable of ranging. The system model therefore assumes a distinct external position-sensing subsystem. Its accuracy, update rate, latency, transport, energy, and coordinate calibration are explicit inputs rather than properties inferred from audio reception.

SDK clock-drift-compensation documentation describes software sample insertion/removal with interpolation, together with PLL adjustment where the board and codec clock topology permit it. PLL operation requires suitable codec-clock control; the documented EVK1.4 fixed-crystal arrangement cannot use that approach~\cite{spark}. Stable queue occupancy or average sample rate is still insufficient to prove common electrical output phase. A timing observation model, clock-domain conversion, and DAC/SAI actuation model enter the feasibility account as distinct conditions. Fig.~\ref{fig:architecture} separates documented broadcast data service from the observation, sensing, and output-actuation prerequisites of acoustic alignment.

\subsection{From documented interfaces to an executable timing path}
The reference timing path attaches a sample-block identifier and source-domain presentation time to each audio block, then maps the presentation time to the corresponding local-clock deadline. The source/receiver clock relation is a state-estimation problem; converting a corrected deadline into a particular SAI/DAC sample is a separate actuation problem. Receiver completion callbacks cannot silently replace calibrated physical receive events, because interrupt scheduling, buffering, and DMA boundaries create additional latency and jitter. Likewise, a software sample correction that regulates long-term queue occupancy need not align the absolute first output sample across independent boards.

The one-to-many service model includes state beyond the broadcast destination value. Its per-receiver logical state includes block and copy indices, duplicate suppression, bounded playout queues, and loss accounting that includes blocks absent from every callback. Application copies occupy scheduled slots even when all receivers already obtained an earlier copy. If feedback is introduced, a collision-free uplink and its airtime must be reserved; if location updates share the radio, their deadlines and state changes must also be accounted for. An address-space limit must never be interpreted as a demonstrated capacity or a reliable group-size bound.

These distinctions motivate the paper's evidence hierarchy. \emph{Documented} means that a capability appears in the reviewed source or technical manual. \emph{Modeled} means that an observation, delay, or stochastic error is explicitly supplied to the simulator. \emph{Measured} would require a specified board, build, instrument, calibration, and trace. The present analysis uses documented and modeled quantities. Thus, the confidence intervals pertain to the stated stochastic models and not to physical output timing.

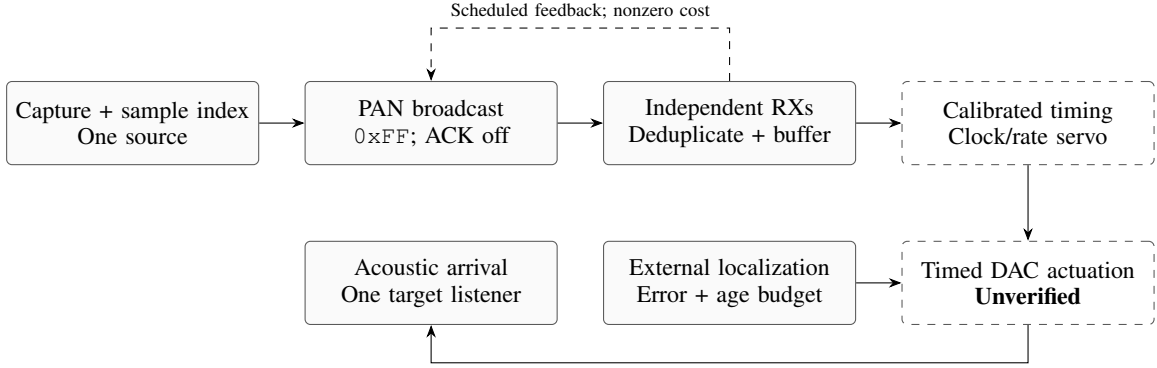
\begin{figure*}[t]
\centering
\begin{tikzpicture}[>=Stealth,node distance=5mm and 6mm,
 box/.style={draw=black!65,rounded corners=2pt,align=center,font=\small,minimum height=11mm,text width=31mm,fill=black!2},
 conditional/.style={box,dashed,fill=white}]
\node[box] (source) {Capture + sample index\\One source};
\node[box,right=of source] (radio) {PAN broadcast\\\texttt{0xFF}; ACK off};
\node[box,right=of radio] (receivers) {Independent RXs\\Deduplicate + buffer};
\node[conditional,right=of receivers] (servo) {Calibrated timing\\Clock/rate servo};
\node[conditional,below=10mm of servo] (dac) {Timed DAC actuation\\\textbf{Unverified}};
\node[box,left=of dac] (location) {External localization\\Error + age budget};
\node[box,left=of location] (listener) {Acoustic arrival\\One target listener};
\draw[->] (source)--(radio);
\draw[->] (radio)--(receivers);
\draw[->] (receivers)--(servo);
\draw[->] (servo)--(dac);
\draw[->] (location)--(dac);
\draw[->] (dac.south)--++(0,-5mm)-|(listener.south);
\draw[->,dashed] (receivers.north)--++(0,7mm)-| node[above,pos=.25,font=\scriptsize]{Scheduled feedback; nonzero cost} (radio.north);
\end{tikzpicture}
\caption{Reference architecture for the conditional multicast-audio model. Solid blocks describe documented transport elements; dashed blocks identify observation and output-control variables represented at the model level. Position measurements come from an external subsystem, and the acoustic objective applies to one designated listener.}\label{fig:architecture}
\end{figure*}

\section{System Model and Analytical Bounds}\label{sec:model}
\subsection{Clock dynamics, units, and observations}
Let $t$ be source time in seconds. Receiver $i$ has local clock
\begin{equation}
 C_i(t)=t+10^{-6}b_i(t),
\end{equation}
where offset $b_i$ is in microseconds and skew $s_i=\dot b_i$ is in microseconds per second, numerically parts per million. At beacon epochs separated by $\Delta$, the reference estimator uses
\begin{align}
 \bm x_{i,k+1}&=\bm F\bm x_{i,k}+\bm w_{i,k},&
 \bm F&=\begin{bmatrix}1&\Delta\\0&1\end{bmatrix},\\
 z_{i,k}&=\bm H\bm x_{i,k}+\beta_i+v_{i,k},&
 \bm H&=\begin{bmatrix}1&0\end{bmatrix},\label{eq:observe}
\end{align}
with $\bm x=[b,s]^T$, fixed observation bias $\beta_i$, and zero-mean timing noise $v$. A calibrated observation is assumed available when a beacon succeeds. It is not identified with a UART host timestamp or an uncalibrated software callback. Equation~\eqref{eq:observe} contains two parameters that contribute additively to the observation: the unknown offset and the fixed one-way timestamp bias. They are not separately identifiable from one-way observations alone; a calibration constraint, reciprocal observation, or independent time reference is required~\cite{sundar,hollosi}.

For a skew random-walk intensity $q_s$, the filter process covariance is
\begin{equation}
 \bm Q=q_s\begin{bmatrix}\Delta^3/3&\Delta^2/2\\\Delta^2/2&\Delta\end{bmatrix}.
\end{equation}
The simulated truth integrates the previous skew over each step and then applies a Gaussian skew increment with standard deviation $\sigma_s\sqrt{\Delta}$. The filter covariance corresponds to an integrated continuous-time approximation; this small deliberate mismatch is disclosed rather than described as a calibrated oscillator model. Bias is retained in the simulated observation, so reducing random measurement noise does not eliminate it. Initial offsets, skews, and per-node biases are independently drawn within the stated ranges.

\subsection{One-listener geometry and causal playback}
Let $\bm p_i$ denote speaker position and $\bm u(t)$ one target listener, both in metres. Acoustic distance is $d_i=\|\bm p_i-\bm u\|$, with estimate $\hat d_i$. The sound speed is $c=343$~m/s in the model; this is distinct from electromagnetic propagation speed used in UWB ranging. For capture time $\tau_k$, fixed buffer $B$, and design bound $D\geq\max_i d_i$, a desired emission time is
\begin{equation}
 t^*_{i,k}=\tau_k+B+\frac{D-\hat d_i}{c}.\label{eq:emit}
\end{equation}
For $C_i=t+10^{-6}b_i$, programming a local deadline $t^*+10^{-6}\hat b_i$ yields first-order firing error $10^{-6}(\hat b_i-b_i)$; reversing this sign would be physically inconsistent. With actuator phase error $q_i$ in microseconds, define
\begin{equation}
 e_i=\hat b_i-b_i+q_i,\qquad
 A_i=\tau_k+B+D/c+10^{-6}e_i+(d_i-\hat d_i)/c.\label{eq:arrival}
\end{equation}
The arrival span is $S_A=\max_i A_i-\min_i A_i$. Deliberately different emission delays may reduce $S_A$; therefore, equal emission timestamps are not always the correct acoustic objective.

Equation~\eqref{eq:emit} is a feasible instruction only if the complete audio block has arrived and the receiver can commit a playback action before the intended emission time. In particular, $B+(D-\hat d_i)/c$ must exceed source-to-queue completion, decoding, clock conversion, and scheduling/actuation delay. The clock simulation is an event-time error model: distances and the most recent sensor estimate are evaluated at the modeled render event, not through a complete future-deadline/audio-buffer implementation. The actuation horizon is represented by prediction or by an additional motion-age term in the bound. The reported scheduled-arrival quantity therefore describes a nominal time allocation in the event model, not observed capture-to-DAC or capture-to-sound latency. With $D=12$~m and $B=10\ms$, the intended common-arrival guard alone is about $44.99\ms$, so no ultra-low acoustic latency is claimed.

\subsection{A conditional acoustic error budget}
\begin{proposition}[Conditional arrival-span budget]\label{prop:budget}
Suppose the clock-estimation error is bounded by $|\hat b_i-b_i|\leq E_c$ in microseconds, $|q_i|\leq q_D/2$, speaker-coordinate error is at most $\delta_p$, listener-estimate error at its timestamp at most $\delta_u$, speed at most $v$, and total effective position age at actuation at most $a$. Then
\begin{equation}
 S_A\leq 2\cdot10^{-6}E_c+10^{-6}q_D+
 \frac{2(\delta_p+\delta_u+va)}{c}.\label{eq:bound}
\end{equation}
\end{proposition}
\begin{proof}
The reverse triangle inequality bounds each distance-estimation error by $\delta_p+\delta_u+va$. In~\eqref{eq:arrival}, every common term cancels under pairwise differences. The range of clock errors is at most $2E_c$, that of the actuator errors at most $q_D$, and that of acoustic distance errors at most twice the individual bound. Summing these ranges and converting microseconds to seconds gives~\eqref{eq:bound}.
\end{proof}

Proposition~\ref{prop:budget} is a sufficient conditional inequality: it does not assert that its bound is tight, that the noise has bounded support, or that any physical clock satisfies the assumed limit. Gaussian sensor/noise variables in the simulation are unbounded and hence do not automatically satisfy its deterministic assumptions. A probabilistic version requires stated-confidence radii and a union bound over the chosen devices and horizon. The 1,000 numerical bounded-error checks supplied with the code test consistency, not the truth of the hardware premises.

The physical meaning of this budget is direct. A 2~cm distance error corresponds to $58.31\us$ one-way acoustic delay. If the electronic residual were zero, a sufficient pairwise $30\us$ budget would allow only $c(30\us)/2=5.145$~mm individual distance uncertainty under the symmetric bound. This is a sufficient worst-case requirement, not a universal localization threshold. Shared position errors may cancel partly depending on geometry. Clock and position statistics must therefore be modeled jointly rather than adding their empirical p95 values and calling the result a p95 guarantee.

\subsection{Correlated broadcast erasures and group reliability}
Let $H_k\in\{g,b\}$ be a common Markov channel state for audio block $k$, with transition probabilities $\alpha=\Prob(g\!\to\!b)$ and $\gamma=\Prob(b\!\to\!g)$. Its stationary bad-state fraction is $\pi_b=\alpha/(\alpha+\gamma)$. This is a reduced-order burst model inspired by established finite-state channel modeling~\cite{gilbert,elliott}, not a fitted UWB propagation model.

For copy $r$, common erasure $E_{kr}$ has probability $q_{H_k}$, shared by all receivers. Receiver-local erasure $U_{kir}$ has probability $p_{iH_k}$. Local losses are independent conditional on the common state and common erasure. The observed loss is $L_{kir}=E_{kr}\lor U_{kir}$; a receiver misses the block if all $R$ copies are lost. For fixed $R$,
\begin{equation}
 m_i(R)=\sum_{h\in\{g,b\}}\pi_h\big[q_h+(1-q_h)p_{ih}\big]^R.\label{eq:missing}
\end{equation}
The probability that all $N$ receivers obtain the block is
\begin{equation}
 P_{\rm all}(R)=\sum_h\pi_h\sum_{j=0}^{R}\binom{R}{j}
 (1-q_h)^j q_h^{R-j}\prod_{i=1}^N(1-p_{ih}^{j}).\label{eq:all}
\end{equation}
To see this, condition on $j$ copies escaping the common erasure; receiver $i$ then fails only if all $j$ local attempts fail. The $j=0$ term is zero for any nonempty receiver population. Mixing over the state gives~\eqref{eq:all}. Time correlation affects bursts and finite-run variability; a stationary marginal formula does not by itself describe those tails. Adaptive policies depending on past states require history-conditioned analysis and are not certified by this fixed-$R$ identity.

An IID estimate $(\sum_h\pi_h\ell_{ih})^R$, where $\ell_{ih}=q_h+(1-q_h)p_{ih}$, generally differs from~\eqref{eq:missing}. For $R>1$, convexity implies $\E[\ell^R]\geq(\E\ell)^R$: ignoring state variation can underestimate individual failure. Cross-receiver common losses also make multiplication of unconditional individual success probabilities invalid. As $N$ grows, even good average receiver delivery may coexist with much poorer all-receiver delivery. This is the reliability quantity relevant to a synchronized group rendering the same block.

\subsection{Group-admission certificates and an unavoidable outage floor}
\begin{proposition}[Fixed-repetition group bounds]\label{prop:groupbounds}
For the stationary channel in~\eqref{eq:all}, with $N\geq1$ and $R\geq1$, let $F_{\rm all}(R)=1-P_{\rm all}(R)$. Then
\begin{equation}
\sum_h\pi_h q_h^{R}\leq F_{\rm all}(R)
\leq\min\!\left\{1,\sum_{i=1}^{N}m_i(R)\right\}.
\label{eq:groupbounds}
\end{equation}
Hence $\sum_h\pi_hq_h^R>\epsilon_g$ rules out a target group-failure probability at most $\epsilon_g$ for that $R$ even if all local links are perfect. Conversely, $\sum_i m_i(R)\leq\epsilon_g$ is a sufficient, generally conservative, group-admission certificate under the stated channel probabilities.
\end{proposition}
\begin{proof}
When every copy is erased in common, no receiver obtains the block, giving the lower bound. Group failure is the union of individual missing-block events, so Boole's inequality and~\eqref{eq:missing} give the upper bound; every failure probability is at most one.
\end{proof}

The two bounds answer different questions. A violation of the common-erasure floor proves infeasibility within the assumed model, but satisfying that necessary condition does not certify group service. The union bound can certify service without assuming independent receiver successes, at the expense of conservatism when many receivers fail together. Exact expression~\eqref{eq:all} is more discriminating when its channel parameters are known. If $\alpha=0.005$, $\gamma=0.1$, $q_g=0.001$, and $q_b=0.2$, then $\pi_b=1/21$ and the unavoidable failure floors for one, two, and three copies are approximately $1.05\times10^{-2}$, $1.91\times10^{-3}$, and $3.81\times10^{-4}$, respectively. These are arithmetic consequences of assumed parameters, not measured packet-loss rates. In particular, a 99.9\% group-delivery target is unattainable with two copies, even with perfect local channels, in this particular assumed model.

A feedback-based algorithm does not inherit these certificates simply by applying the fixed-$R$ formula to an average loss estimate. It would require knowledge or validated uncertainty bounds for the hidden state, conditional correlations between copies, report delays, and errors in its current observations. The reference adaptive policy later evaluated has none of these guarantees. We use the certificate as an analysis tool and a design acceptance condition, not as evidence of an implemented state-aware controller.

\subsection{Causality, observability, and coupled feasibility}
Define $t_{i,k}^{\rm ready}$ as the source-time equivalent of the instant at which receiver $i$ has a complete, decodable audio block and can commit its output action; let $\tau_{i,k}^{\rm exec}$ be the additional required execution lead time. A necessary condition for the target emission~\eqref{eq:emit} is
\begin{equation}
 t_{i,k}^{\rm ready}+\tau_{i,k}^{\rm exec}\leq
 \tau_k+B+\frac{D-\hat d_i}{c}.\label{eq:causal}
\end{equation}
The clock and transport simulators do not calculate $t_{i,k}^{\rm ready}$ or $\tau_{i,k}^{\rm exec}$. Consequently, a timing residual computed with~\eqref{eq:arrival} does not establish that~\eqref{eq:causal} is satisfied. Adding repeats late in the frame may improve delivery in an erasure abstraction yet prevent the recovered block from meeting its intended output deadline.

The conditional system feasibility region is the intersection of four requirements: (i) identifiable clock observations and a controllable output actuator; (ii) a specified upper bound or quantile for the one-listener acoustic span; (iii) group-level block delivery meeting its risk tolerance; and (iv) sufficient reserved service within the required interval, including beacons, feedback, and sensing traffic. The first requirement is a hardware prerequisite; the others are model-level constraints. Passing any one does not imply the remainder. In particular, the upper bound of Proposition~\ref{prop:budget} is sufficient for timing under bounded errors, while the common-erasure floor in Proposition~\ref{prop:groupbounds} supplies a necessary reliability check. These logical directions must not be conflated.

For the illustrative $30\us$ span threshold, a single-sample actuation quantum of $20.833\us$ leaves at most $9.167\us$ of the worst-case sufficient bound for clock and geometry. If clock errors were identically zero, the remaining symmetric distance-uncertainty allocation would be about $1.57$~mm per speaker-to-listener path. This follows from~\eqref{eq:bound}; it is not an experimentally established localization requirement. It demonstrates why a centimetre-scale position-error assumption cannot support a tens-of-microseconds deterministic arrival guarantee merely by refining the radio clock.

\subsection{Resource constraints and the design objective}
For $L$ SWC payload bytes, the documented SR1120 IOOK FEC1.25 model at 20.48~MHz gives~\cite{spark}
\begin{equation}
 a(L)=(15.97+0.4883L)\us.\label{eq:airtime}
\end{equation}
For $L=112$, $a=70.6596\us$. The modeled reserved slot $s=150\us$ additionally assumes processing, SPI, concurrency checks, and margin fit; it is not a measured minimum. Let block period be $T_f$, beacon reserve $s_b/T_b$, and per-receiver reporting reserve $s_{\rm fb}/T_{\rm fb}$. Offered schedule load is
\begin{equation}
 \rho=\frac{\E R\,s}{T_f}+I_{\rm adapt}\frac{Ns_{\rm fb}}{T_{\rm fb}}+\frac{s_b}{T_b}.\label{eq:load}
\end{equation}
Worst-case admission uses the maximum allowed $R$, not only its average. When $\rho>1$, delivery outcomes from the erasure process are not valid real-time service predictions. We mark these points infeasible rather than pretending the queue can remain stable.

A prospective controller could minimize $\E[R]a+\lambda\E[T_{\rm completion}]$ subject to group failure, arrival spread, playback deadline, and scheduling constraints. This expression identifies competing quantities; neither an optimal policy nor a closed-loop scheduler is solved by the present simulation. The current policies are transparent reference heuristics. Source data airtime is $\E[R]a/T_f$, while aggregate audio-receive activity is approximated by $N\E[R]s/T_f$. The latter is a duration proxy, not measured energy. A complete energy accounting boundary includes receiver feedback, beacons, idle/wakeup, DSP, codec, and positioning costs.

\section{Reference Algorithms and Implementation Boundaries}\label{sec:methods}
\subsection{Affine clock tracking and playback mapping}
At each beacon, a standard Kalman predictor forms $\hat{\bm x}^-=\bm F\hat{\bm x}$ and $\bm P^-=\bm F\bm P\bm F^T+\bm Q$. On a valid observation,
\begin{align}
 \bm K&=\bm P^-\bm H^T(\bm H\bm P^-\bm H^T+\sigma_v^2)^{-1},\\
 \hat{\bm x}&=\hat{\bm x}^-+\bm K(z-\bm H\hat{\bm x}^-),\\
 \bm P&=(\bm I-\bm K\bm H)\bm P^-(\bm I-\bm K\bm H)^T+\bm K\sigma_v^2\bm K^T.
\end{align}
The Joseph covariance update preserves positive-semidefinite structure numerically. A lost beacon triggers prediction only. For a render event $h$ seconds after the last beacon, the offset estimate is $\hat b(h)=\hat b_k+h\hat s_k$. Receiver state and matrix work are constant per node; network-wide update cost and state scale as $O(N)$. This is a reference use of established estimation theory~\cite{kalman}, not an estimator novelty claim.

Four paired timing policies use identical simulated observations and geometry: initial-offset-only free running; last-valid-offset correction without skew estimation; affine correction without position; and affine correction with the most recent external position estimate. Actuation is ideal except for an independent uniform event-phase error on $[-q_D/2,q_D/2]$. This models a granularity floor and is not a waveform resampler, PLL controller, correlated DAC jitter process, or proof that an embedded implementation can move playback phase freely. Rate correction must in reality preserve continuous audio and avoid queue underflow, audible discontinuities, and shared-PLL conflicts~\cite{didier,korse,spark}.

\subsection{Repetition policies and feedback accounting}
The fixed policies transmit one, two, or three application copies with a shared block sequence and distinct copy identifiers. Receivers deduplicate copies and consume each block at most once. Candidate loss outcomes are drawn once per run and reused across policies, ensuring paired comparisons on the same channel realization. Policies never use an outcome from a copy that would not have been transmitted.

Every $T_{\rm fb}=0.2$~s, an adaptive policy receives a per-receiver first-copy loss fraction $y_i$ from the preceding interval and updates $\hat p_i\leftarrow0.7\hat p_i+0.3y_i$, with initial value 0.1. Using only first-copy outcomes avoids selection bias from observing extra copies only after choosing them. Mean-risk aggregates $\bar p=N^{-1}\sum_i\hat p_i$; worst-risk uses $\max_i\hat p_i$. For target proxy $\epsilon=0.01$,
\begin{equation}
 R=\min\!\left(3,\max\!\left(1,\left\lceil\frac{\log\epsilon}{\log p_{\rm agg}}\right\rceil\right)\right).\label{eq:heuristic}
\end{equation}
The implementation clips estimated probabilities to $(10^{-8},1-10^{-8})$ for numerical stability. Equation~\eqref{eq:heuristic} extrapolates an IID per-copy failure probability to a repetition count. It neither computes group failure nor represents common bursts, finite-sample uncertainty, delayed feedback, or a binding service limit. The nominal $1\%$ parameter is therefore a heuristic tuning value, not a reliability guarantee. Dedicated reports cost $N$ slots per reporting interval and are assumed received correctly. Their RF airtime is not equated with the entire reserved-slot duration. Fixed-copy baselines do not pay for unused reports, preventing an artificial adaptive advantage.

\begin{algorithm}[t]
\caption{Reference error-budget and repetition procedure}\label{alg:reference}
\begin{algorithmic}[1]
\State Initialize $\hat{\bm x}_i,\bm P_i,\hat p_i=0.1$ for each receiver
\For{each beacon epoch}
 \State Predict $\hat{\bm x}_i,\bm P_i$ for all $i$
 \State Update only nodes with valid calibrated observations
\EndFor
\For{each modeled render event}
 \State Extrapolate offset and read latest position plus age
 \State Form intended emission delay using~\eqref{eq:emit}
 \State Evaluate electronic and one-listener arrival spans
\EndFor
\For{each reporting interval}
 \State Update $\hat p_i$ using past first-copy observations only
 \State Select fixed $R$ or apply~\eqref{eq:heuristic}
 \State Evaluate mean and maximum schedule load
 \State Mark overload outcomes infeasible for real-time claims
\EndFor
\For{each audio block}
 \State Send $R$ application copies; deduplicate at receivers
 \State Account per-receiver and simultaneous group delivery
\EndFor
\end{algorithmic}
\end{algorithm}

\subsection{Offline feasibility audit versus online policies}
The analytical admission test differs deliberately from the implemented policies. If the state probabilities and receiver-local loss parameters are specified, an \emph{offline} analyst can evaluate~\eqref{eq:all} for each permissible fixed $R$ and reject a configuration whenever $P_{\rm all}(R)<1-\epsilon_g$. The common-erasure floor~\eqref{eq:groupbounds} offers an immediate impossibility check, and the union-bound certificate offers a conservative sufficient check. These evaluations do not require estimating an unobserved state from live receiver reports; they require accepting the model parameters as inputs. They are not a new feedback strategy, and they do not substitute for calibration.

The corresponding online decision is substantially harder. A causal controller sees a delayed sample of first-copy outcomes rather than the hidden common channel state; feedback may itself be corrupted, reporting intervals consume reserved slots, and execution deadlines constrain how late another copy can be sent. Thus a theoretically preferable repetition decision based on perfect current state would be an oracle baseline, not an implementable guarantee unless its sensing and control costs are included. We distinguish fixed-copy feasibility, offline model-based assessment, and the two deployed-in-simulation IID heuristics when interpreting results.

The estimator also has a structural blind spot: a constant receive-side observation bias may be absorbed into the estimated clock offset. The affine state can track the apparent offset without recovering the unbiased offset needed for absolute output timing. Clock-observation bias and output actuation therefore enter the feasibility model independently. A reduced estimator covariance is not proof that the true timestamp or emitted audio is equally accurate. The model intentionally retains a bounded observation-bias term to keep this distinction visible.

\subsection{Implementation boundaries and reproducibility}
Algorithm~\ref{alg:reference} describes the reference components and validity checks; the simulation executes clock and transport components separately. It does not run a coupled codec/playout queue, network discovery procedure, admission controller, source-authentication protocol, or location-zone handover. Marking overload is an analysis validity condition, not an embedded scheduler implementation. The $O(N)$ reporting update does not include the potentially dominant cost of obtaining calibrated position and timing data.

The supplied Python package includes configuration, simulation, run-level analysis, plotting, analytical checks, acquisition utilities, and unit tests. A post-run sign audit corrected the local-deadline firing error to $\hat b-b$; the archived statistics and plots correspond to the corrected model. Current outputs carry source and configuration hashes. Seventeen tests check numerical identities, Joseph covariance properties, repeatability, shared-loss behavior, deadline sign, complete transmission denominators, duplicate handling, and rejection of incomplete hardware metadata. The original execution harness depended on an environment-specific process supervisor and absolute paths. The distributed package supplies relative-path verification commands and retains the computational model and its archived outputs. Passing tests establish internal consistency, not calibrated external validity.

\section{Experimental Design and Evaluation Protocol}\label{sec:evaluation}
\subsection{Operating point and parameter provenance}
The offered stream is 48~kHz stereo, 16~bits per sample, divided into 0.5~ms blocks: 24 stereo sample frames correspond to 96~PCM bytes, with a further 16-byte assumed application header. The 112-byte packet is a custom transport assumption, not a drop-in modification of the stock SAC pipeline. Actual SAC header bytes, fragmentation, and feature overhead must be budgeted separately. No compression algorithm or audio waveform is simulated; successful block reception represents intact offered PCM payload only.

Table~\ref{tab:parameters} distinguishes documented coefficients from assumptions. Channel transition probabilities, independent/local erasure probabilities, clock and sensor statistics, actuation granularity, and reserved-slot durations are modeling assumptions rather than fitted hardware measurements. The objective is sensitivity and failure-boundary exploration, not numerical prediction of an uncalibrated board. FEC1.25 affects the documented airtime coefficient only; no unmeasured coding-gain/PER curve is invented. The separate 10~ms beacon reserve is deliberately conservative even if frequent data transmissions could maintain radio synchronization. An actual schedule must also show that feedback slots fit at their intended epochs: utilization at or below one is necessary but not sufficient for discrete slot-placement and deadline feasibility.

\begin{table}[t]
\caption{Main parameters. Except documented airtime, these are modeling assumptions, not hardware measurements.}\label{tab:parameters}
\centering\footnotesize
\begin{tabularx}{\columnwidth}{Yp{30mm}}
\toprule
Parameter & Value \\
\midrule
Audio / SWC bytes per block & 48~kHz stereo16;112~B \\
Block / assumed reserved slot & 0.5~ms /150~$\mu$s \\
SR1120 airtime coefficients~\cite{spark} &15.97~$\mu$s;0.4883~$\mu$s/B \\
Beacon period / reserved slot &10~ms /100~$\mu$s \\
Feedback period / slot per RX &200~ms /100~$\mu$s \\
Maximum copies / risk proxy &3 /0.01 \\
Clock initial offset / skew &uniform$\pm500$~$\mu$s /$\pm30$~ppm \\
Skew walk / timestamp noise &0.05~ppm/$\sqrt{\rm s}$ /2~$\mu$s SD \\
Timestamp fixed bias &uniform$\pm2$~$\mu$s \\
Actuator quantum / render step &20.833~$\mu$s /5~ms \\
Position coordinate SD / update &0.02~m /100~ms \\
Nominal position age / max speed &50~ms /0.5~m/s \\
Common / local beacon loss &0.01 /0.03 \\
Transport state transitions $\alpha,\gamma$ &0.005,0.1 per block \\
Common erasure $q_g,q_b$ &0.001,0.2 per copy \\
Local erasure $p_{ih}$ &$0.003+0.06(r_i/10)^2+0.12\mathbf1_b$ \\
Transport radius $r_i$ &uniform1--10~m \\
Timing speaker radius &uniform2--8~m \\
Sound speed / $D$ / $B$ &343~m/s /12~m /10~ms \\
\bottomrule
\end{tabularx}
\end{table}

\subsection{Trial matrix and statistical estimands}
Main timing experiments use eight receivers, 30~s runs, 10~ms beacons, 5~ms render events, and a 2~s excluded warmup. Main transport experiments use $N\in\{2,4,8,16,32,64\}$, 10~s per run, and 20,000 blocks per run. Each setting uses 30 seeds starting at 20261017. Timing sensitivity runs last 10~s with the same warmup; their results are not pooled with 30~s main runs.

Sensitivity covers position coordinate standard deviations $\{0,0.005,0.02,0.1,0.5\}$~m; timing noise $\{0.5,2,10,50\}\us$; actuator quanta $\{0.5,20.833\}\us$; nominal position ages $\{0,0.05,0.2,0.5\}$~s; and six receiver counts. Transport sensitivity covers three bad-state common-erasure probabilities at $N=16$ and four reserved-slot lengths at $N=64$. This yields 120 main timing, 900 main transport, 2,520 timing-sensitivity, and 1,050 transport-sensitivity policy--run records. The 4,590 records count policy--run outcomes, not independent physical measurements: multiple policies share a random seed and, within each experiment family, relevant underlying draws.

The primary timing estimand is the arithmetic mean of 30 run-specific 95th percentiles. It is not the 95th percentile of events pooled across all 30 runs, and it need not equal that pooled quantile. For a run-level metric $X_j$, the plotted interval is
\begin{equation}
 \bar X\pm t_{0.975,29}\,s_X/\sqrt{30}.
\end{equation}
Intervals describe variability under the fixed assumed model, not uncertainty about product performance or all possible rooms. Packet samples within one run are not treated as independent replications. No multiplicity-adjusted null-hypothesis tests are carried out: interval widths summarize Monte Carlo variation at fixed parameter values, and the analysis emphasizes interpretable effect sizes and failure modes. Comparisons across $N$ do not enforce nested identical geometries, so small nonmonotonic fluctuations are not physical capacity effects.

\subsection{Separate observables and coupled failure criteria}
The simulation comprises two intentionally separate families. The clock experiment updates the estimator at beacon epochs and evaluates receiver emissions at event times while supplying an idealized output-phase error and sampled external positions. The transport experiment draws state-dependent common and receiver-local erasures for up to three copies of each block, then computes delivery and reservation metrics. Neither experiment integrates the other subsystem's outcomes into a continuous audio queue. In particular, blocks classified as received in the transport experiment are not inserted into a PCM decoder or scheduled against the emission deadlines of the clock experiment.

This separation is useful for identifying individual physical bottlenecks but prevents an end-to-end success probability from being estimated by simply multiplying reported timing and delivery statistics. The correlation between a receiver missing a block, waiting for later copies, losing a timing observation, and missing a DAC deadline is unspecified. Under a deployed system, these events can share interference, buffering, and processor-service mechanisms. Their joint law cannot be recovered from the marginal tables in this paper. We therefore report the separately identifiable feasibility conditions without assigning an unsupported joint user-quality probability.

The analytical checks added to the manuscript use the same fixed common-state parameters as the transport model but are not additional measured trials. We separately evaluate the common-erasure failure floor for $R\in\{1,2,3\}$ and the nominal reserved-load threshold implied by the configured slot duration. These computations are deterministic consequences of parameter choices and can be verified without an RF stack. They complement, but do not calibrate, the Monte Carlo experiments.

\subsection{Cross-method comparison and statistical limitations}
The timing policies share the same within-run clock and geometry realizations; the transport policies share the same potential copy losses within a run. Consequently, differences within a seed are paired and should not be interpreted as estimates from unrelated physical deployments. Between seeds, receiver locations and channel paths vary according to the chosen assumptions. We summarize each condition with a run-level mean and confidence interval, retaining the seed and policy labels in the raw CSV. Inference beyond the modeled geometry, mobility trajectory, or erasure probabilities requires a second layer of environment-level or session-level sampling that is unavailable here.

Failure-regime sweeps are intentionally stress tests. A parameter combination that violates the utilization bound is not evidence of a stable overloaded audio implementation; it demonstrates that the independent packet-erasure results cannot be scheduled using the assumed time slots. Likewise, a low residual phase range sampled independently at each render event is not a demonstrated continuous-time clock servo. Both caveats are treated as validity conditions before performance comparisons are made.

\subsection{Metrics, baselines, and validity conditions}
Electronic residual span is $\range_i(e_i)$, distinct from the total emission span after intentional distance delays. Arrival span follows~\eqref{eq:arrival}. We report per-run p95 and maxima, plus the fraction above the illustrative $30\us$ threshold. The timing geometry uses speakers on independently drawn angles and radii and a bounded-speed listener trajectory $\bm u(t)=[(0.5/0.2)\sin(0.2t),0]^T$~m. Position estimates are refreshed every 0.1~s from the trajectory 0.05~s earlier with shared two-dimensional Gaussian coordinate noise; speaker positions are exact. Consequently, nominal age understates event-time age by up to approximately one update interval, before any real actuation horizon.

For transport, receiver delivery is $D_i=1-K^{-1}\sum_k M_{ki}$. We report average delivery, $\min_iD_i$, simultaneous all-receiver delivery, maximum consecutive missing-block duration, and the index $J=(\sum_iD_i)^2/(N\sum_iD_i^2)$. The latter is named a Jain-style delivery index and used as a descriptive statistic, not as proof of application fairness. A shared outage can produce perfect equality and unusable service. Lost blocks are not automatically buffer underruns or perceptual dropouts because no concealment, codec, or coupled playback buffer is simulated.

Fixed-copy and risk-based transport methods are compared under the same candidate channel outcomes, with their actual repetition and report reservations reported. This is not an equal-energy policy tournament: fixed-three spends more airtime and adaptive methods pay feedback. Policy ordering under unequal repetitions is not evidence of a Pareto improvement; optimization claims are conditional on the delivery target, presentation deadline, and complete feedback, sensing, and energy accounting boundary. Overloaded points expose a boundary but do not supply valid delivered-audio performance. Likewise, source airtime and receiver-state duration are kept separate from unmeasured watts, joules, and battery life.

\subsection{Reproducibility, Data Integrity, and Claim Traceability}\label{sec:reproducibility}
The computational record retains the configuration, implementation, per-run observations, and summary estimands as separate artifacts. The executable model is provided in \texttt{simulation.py}; \texttt{config.json} specifies trial durations, clock and position parameters, common-loss transition probabilities, copy limits, resource reservations, and random seeds. The analysis script derives the tables and graphical summaries from saved trial outputs instead of maintaining the full time-step history in memory. Clock and transport experiments are stored in separate tabular outputs, together with sensitivity records and summaries. These separations permit every reported metric to be traced to a defined policy, parameter setting, and run identifier.

The verification procedure distinguishes two operations. An integrity check assesses archived outcomes without repeating the computational sweep: expected row counts, seed coverage, nonempty outputs, input-source hashes, and consistency of numerical claims are inspected. An independent execution follows the same configuration and evaluation code to reproduce the numerical study. The two operations answer different questions. Matching an output hash establishes artifact integrity under a particular source/configuration identity; it does not by itself establish correctness of the mathematical model. Likewise, rerunning a stochastic program with the documented seeds establishes reproducibility of its calculations, not representativeness of its external input distributions.

A policy--run record is the observational unit for confidence intervals and is not equivalent to an independently sampled radio deployment. Within a run, several candidate methods observe the same potential fading and clock events, which makes comparisons paired. For independent runs, seeds determine different paths and geometries according to the model. This design avoids inflating statistical precision by treating temporally correlated blocks as independent replications. It also allows a policy difference to be reproduced using the same candidate random events instead of comparing two unrelated channel samples.

Several interpretive invariants are checked alongside the numerical outcomes. Source payload airtime follows the documented airtime expression for the selected payload length; the reserved-slot coefficient is accounted for separately. Receiver activity is a sum of modeled active durations, not a current measurement. All-receiver block delivery uses the source manifest as its denominator, avoiding the upward bias that results from counting only blocks that triggered a receive callback. A block received by a subset of nodes is never counted as full-group delivery. In the timing model, acoustic arrival spread is computed after geometric compensation and is not relabeled electronic clock jitter. These invariants align the implementation, analytical definitions, and reported figures.

Finally, the paper's computed quantities are divided into conditional mathematical results, Monte Carlo estimates, and documented interface parameters. Conditional bounds use their stated channel and error premises; Monte Carlo confidence intervals refer to the corresponding simulation configuration; and air-interface coefficients are cited to the vendor documentation where applicable. No algebraic identity is inferred from a confidence interval, and no sample-duration proxy is interpreted as power or battery lifetime. This three-part evidence classification makes the numerical conclusions independently auditable at their declared level of abstraction.

\section{Simulation Results and Analysis}\label{sec:results}
\subsection{Clock ablation: a small electronic gain is not acoustic success}
Fig.~\ref{fig:clock} and Table~\ref{tab:clockresults} show the timing ablation. Initial-offset-only free running accumulates skew: its mean run-p95 electronic span is $1341.47\us$. Offset-only resetting reduces this to $23.784\us$, and affine tracking gives $21.186\us$. The modest affine improvement is expected when beacons arrive every 10~ms and a $20.833\us$ actuator granularity dominates the residual. Under this model, a sophisticated estimator alone does not remove the sampling/actuation floor.

The acoustic result differs by orders of magnitude. Without distance correction, the affine policy has a mean run-p95 arrival span of $21953.66\us$. Position compensation reduces this to $419.43\us$ with a 95\% interval half-width of $8.37\us$. This is a substantial model-level improvement, yet 98.864\% of evaluated events remain above the $30\us$ study target. Both positive and negative statements matter: propagation correction helps, but the assumed sensing and update system does not establish the desired alignment.

The two affine policies have exactly the same electronic residuals because they differ only in distance compensation. That equality is an intended ablation, not independent evidence confirming hardware repeatability. Similarly, the computed common-arrival schedule has a roughly 45~ms guard from $B+D/c$; the observed correction cannot be used to claim lower wireless transport latency than a commercial technology. The model estimates different stages for different purposes and does not collapse them into one favorable latency headline.

\begin{table}[t]
\caption{Main timing results, $N=8$. Mean of run-level p95, with 95\% confidence half-width across 30 runs; all values simulated.}\label{tab:clockresults}
\centering\small
\begin{tabularx}{\columnwidth}{Yrr}
\toprule
Policy & Electronic ($\mu$s)& Arrival ($\mu$s)\\
\midrule
Free running & $1341.47\pm85.01$&$22098.00\pm1090.94$\\
Offset only &$23.784\pm0.113$&$21953.61\pm1069.22$\\
Affine, no position &$21.186\pm0.154$&$21953.66\pm1069.22$\\
Affine + position &$21.186\pm0.154$&$419.43\pm8.37$\\
\bottomrule
\end{tabularx}
\end{table}
\begin{figure*}[t]
\centering\includegraphics[width=\textwidth]{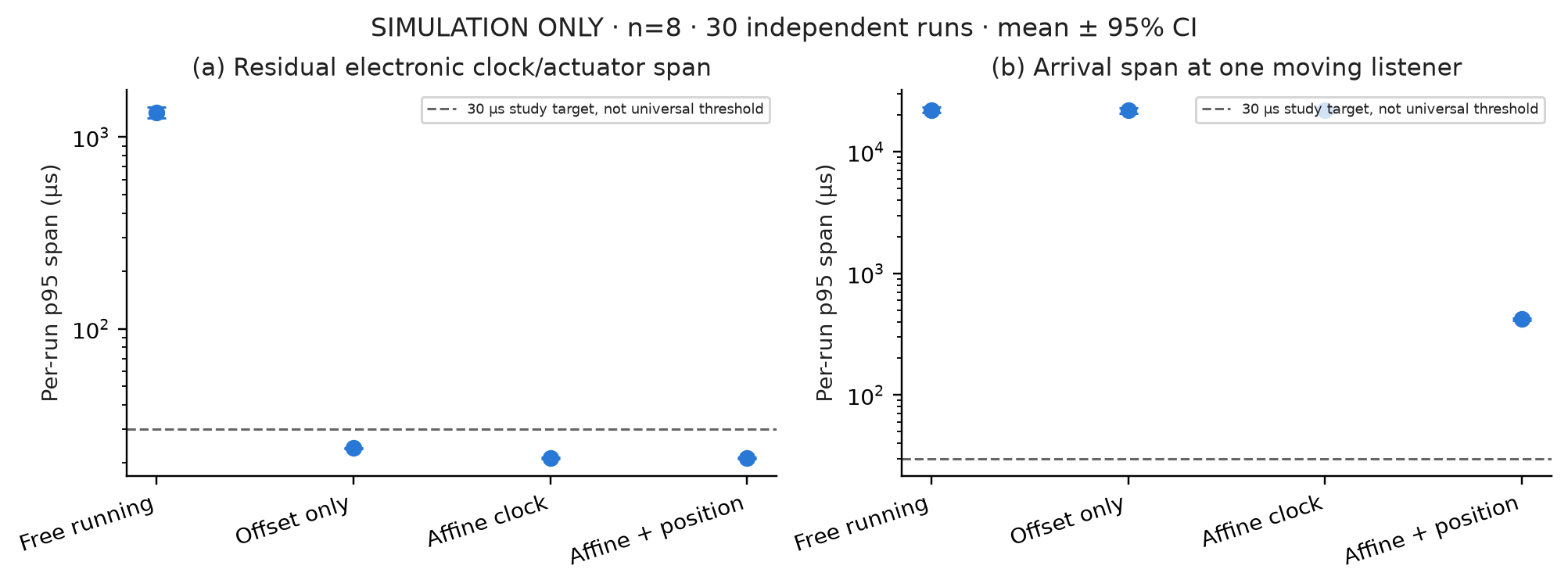}
\caption{Simulation-only clock ablation. Points are means of run-level p95 spans; bars are 95\% $t$ intervals over 30 independent runs. Logarithmic axes separate the electronic residual from one-listener arrival spread. The dashed $30\us$ line is an illustrative study target, not a universal perception threshold.}\label{fig:clock}
\end{figure*}

\subsection{Position age, geometry, and actuator sensitivity}
Fig.~\ref{fig:position} isolates sensor and actuator limitations using the shorter sensitivity runs. Even zero coordinate noise does not make compensated arrivals sample-aligned because motion and stale estimates remain. At the baseline 0.05~s nominal position age, the mean run-p95 compensated span is about $288.59\us$ with zero coordinate noise and about $382.90\us$ with 2~cm coordinate noise in this 10~s setting. These values differ from the 30~s main result because the trajectory segment, random realizations, and horizon differ; the difference is not suppressed by pooling.

Increasing nominal sensor age from 0 to 0.5~s increases the compensated sensitivity span from approximately $305.46\us$ to $1425.60\us$ in the final CSV. The decisive mechanism is motion during information age. Lowering actuation granularity from one sample to $0.5\us$ sharply reduces the electronic residual but barely changes the position-limited arrival span. It is therefore unjustified to optimize RF timestamp resolution while ignoring a 100~ms sensor update cycle. The budget in~\eqref{eq:bound} predicts this qualitative bottleneck without claiming calibrated sensor accuracy.

Geometry also affects scalability. The position error is a shared coordinate perturbation, not an independent distance error injected separately at each receiver; angular configuration changes how it projects onto distances. More receivers can enlarge the maximum-minus-minimum spread even if per-node timing statistics remain unchanged. The model holds speaker coordinates exact and sound speed fixed, making the result optimistic relative to uncalibrated speakers, temperature variation, room multipath, and head movement. No claim of whole-room or multi-listener phase correction follows.

\begin{figure*}[t]
\centering\includegraphics[width=\textwidth]{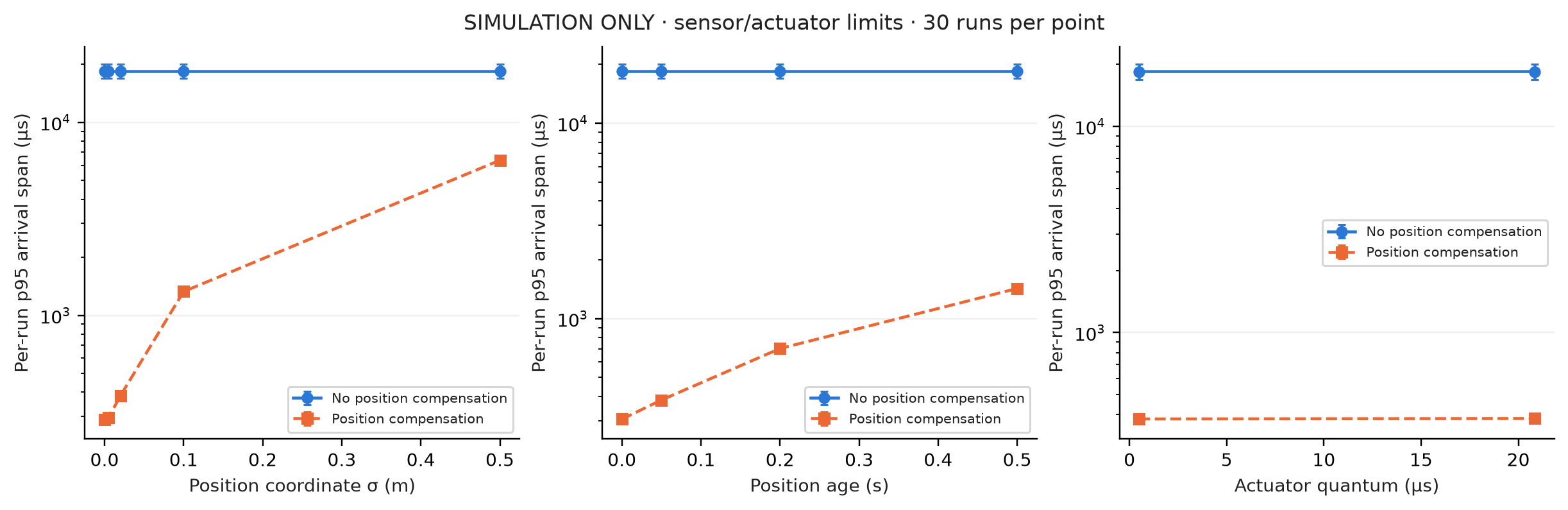}
\caption{Simulation-only position-error, sensor-age, and actuator-quantum sensitivities. Thirty 10~s runs per setting, with 2~s warmup, not pooled with the 30~s main timing study. Error bars are run-based 95\% intervals. Position and timing parameters are assumed; the architecture targets a single listener.}\label{fig:position}
\end{figure*}

\subsection{Receiver reliability versus simultaneous group delivery}
Fig.~\ref{fig:reliability} distinguishes the average receiver, worst receiver, and all-receiver event. At $N=64$, one copy gives mean receiver delivery of about 0.96027 in the model but all-receiver delivery only 0.18522. Reported values come from the archived transport-run CSV and are not inferred from another platform's advertised specifications. Three copies improve worst-receiver delivery to 0.997517 and all-receiver delivery to 0.974777, with respective 95\% half-widths 0.000112 and 0.001037. A user-facing claim that ``delivery is above 99\%'' would conceal the group-level criterion and therefore be misleading.

Fixed two copies yield worst-receiver delivery about 0.99033 and all-receiver delivery about 0.90398. Mean-risk selects two copies throughout the main setting, reproducing those delivery outcomes while paying for reports. Worst-risk uses about 2.011 copies on average at $N=64$ and gives all-receiver delivery about 0.90473, a small effect at higher control cost. These outcomes do not justify a universal adaptive-policy advantage. Fixed three buys reliability through a visibly larger resource allocation rather than winning at the same budget.

The trend is consistent with~\eqref{eq:missing}--\eqref{eq:all}: receiver heterogeneity and common fading prevent population growth from being treated as free reliability scaling. The auxiliary analytical checker also illustrates an IID surrogate with matched marginal losses, although its receiver-location draw order is not the same as in the transport simulation. It is used to check the formula and illustrate dependence, not to claim exact prediction of each simulated row. No curves are labeled measured RF performance or Auracast packet loss.

\begin{figure*}[t]
\centering\includegraphics[width=\textwidth]{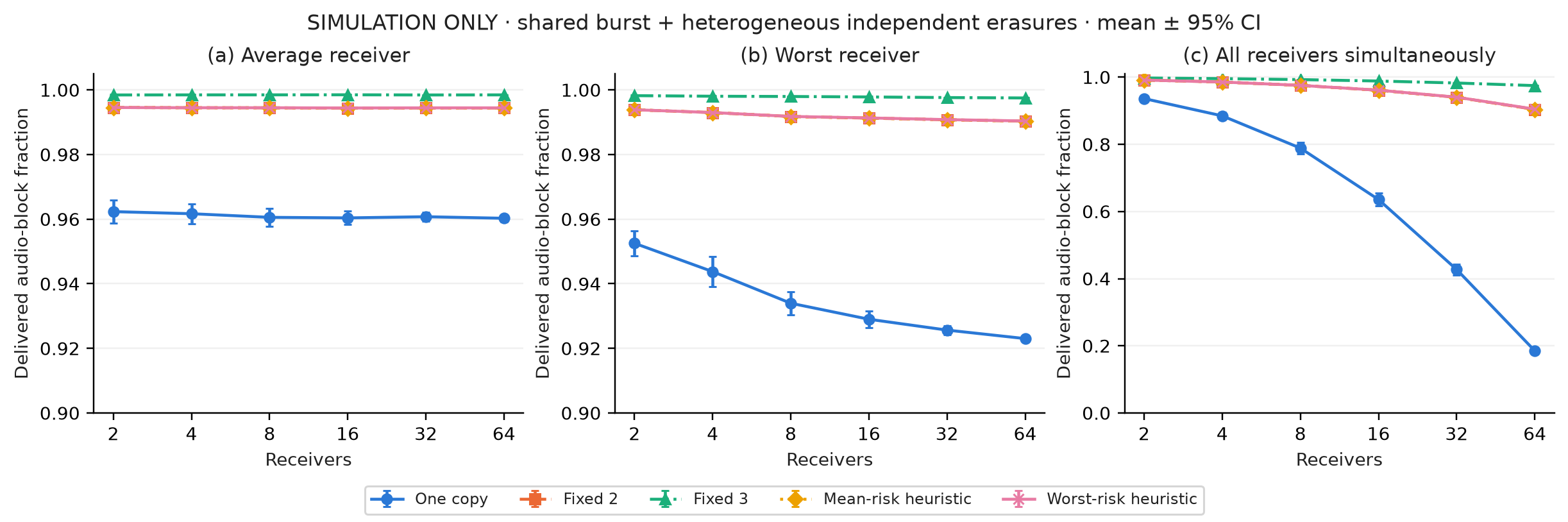}
\caption{Simulation-only delivery scaling under common Markov fades and heterogeneous local erasures. Error bars use 30 independent runs per point. Panels (a,b) zoom the delivery axis to 0.9--1.005; panel (c) uses 0--1.01. Simultaneous group delivery is much stricter than average individual delivery.}\label{fig:reliability}
\end{figure*}

\subsection{Analytical group-risk floors and their interpretation}
Proposition~\ref{prop:groupbounds} permits a channel-model check that does not depend on the sampled receiver positions. For the main state transitions, the stationary bad-state fraction is $1/21\approx0.04762$. When every copy is erased by the common mechanism, all receivers fail regardless of the number of listening devices or their local channel quality. The corresponding unavoidable group-failure probabilities for $R=1,2,3$ are $0.0104762$, $0.00190571$, and $0.000380953$. These bounds decrease with repetition but do not account for any receiver-local erasure and hence systematically understate the actual group-failure probability whenever such losses occur.

This comparison clarifies why a per-receiver risk target cannot be transferred directly to the complete group. In the 64-receiver, three-copy setting, the simulated group failure is approximately $1-0.974777=0.025223$, much larger than the common-erasure floor alone. The gap is not a contradiction: failures of one or more heterogeneous local links add to events in which all copies are lost in common. Conversely, a common-outage process can keep mean receiver delivery high while causing simultaneous failure for many users. An adaptive rule based only on mean first-copy loss cannot distinguish these cases from marginal loss observations alone.

The bound is also an admission test rather than a forecast of temporal burst length. A stationary single-block probability does not determine the duration of consecutive missing blocks or the recovery time of a buffering system. The simulations retain longest missing-burst statistics for this reason, while the absence of an audio queue prevents those bursts from being relabeled as audible dropouts. The combination of group delivery, burst duration, and deadline accounting is needed to decide whether extra copies are useful under a specified playback objective.

\subsection{Airtime, activity, and the limits of fairness summaries}
Fig.~\ref{fig:resources} separates four quantities that are often conflated. With fixed repetition, source data airtime is receiver-independent in this abstraction. Three copies of 112~bytes consume about 423.96~ms of source data airtime per second, while their reserved slots consume 900~ms/s before the separate beacon reserve. Source data airtime excludes feedback and is lower than reserved service time because the latter includes assumed non-radio work and margin.

The aggregate receiver audio-on proxy scales linearly with $N$. For 64 receivers and three copies it is 57,600 receiver-ms/s: the sum of 900~ms/s over 64 receivers, not a physically impossible duty cycle of one radio. Each receiver's energy additionally depends on hardware state currents, codec/DSP activity, feedback, synchronization, and localization. We deliberately provide no invented milliwatt or battery-lifetime values. The word ``free'' is appropriate only to the narrow marginal sender-data-airtime property for fixed coding/repetition, not to total system cost.

The delivery fairness index is close to one for all policies. For fixed-three at $N=64$, it is approximately 0.99999988, yet all-receiver delivery is only 0.97478. An outage affecting every receiver equally gives perfect equality but poor service. Accordingly, fairness should be reported with worst-receiver reliability, shared-outage frequency, service admission, and user-level utility rather than as a standalone success criterion. Location-based zoning would introduce an additional distinction between intentionally excluded and unfairly underserved users; that service layer is not implemented here.

\begin{figure*}[t]
\centering\includegraphics[width=0.92\textwidth]{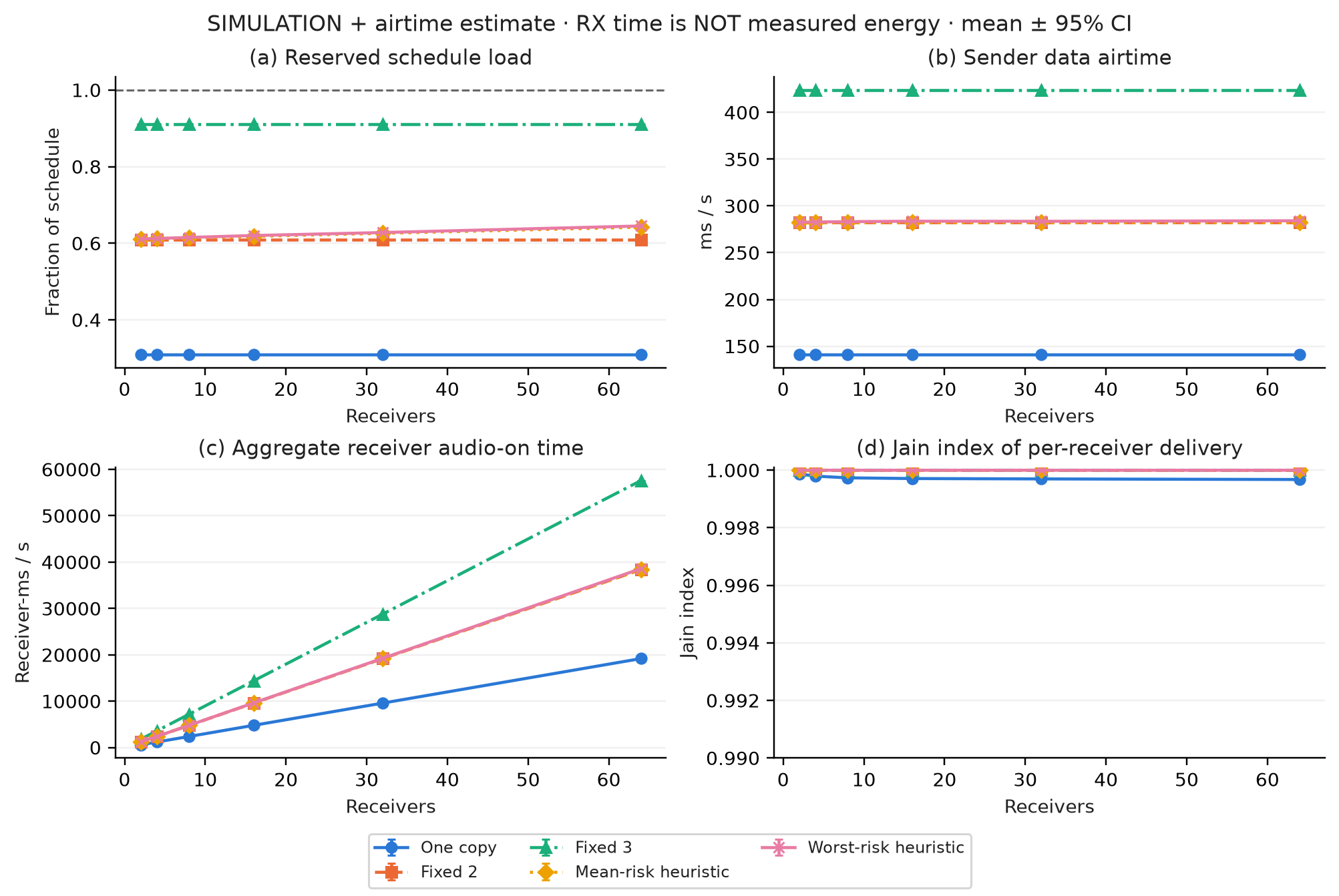}
\caption{Simulation and documented-airtime estimates. Reserved schedule load is not on-air duration; summed receiver activity is not measured energy. Jain-style delivery equality can remain high despite common failures. Error bars summarize 30 runs. No current, power, or lifetime measurements are present.}\label{fig:resources}
\end{figure*}

\subsection{A shared timing--reliability feasibility boundary}
The two simulation families reveal different reasons a group service can miss its objectives, and neither compensates for the other. In the baseline timing model, the single-sample output-phase quantum already consumes $20.833\us$ of the sufficient worst-case $30\us$ span allowance. A centimetre-scale location error or motion over a stale sensing interval can exhaust the remaining margin without any increase in radio packet loss. In the baseline 64-receiver transport model, fixed-three improves complete-block delivery but reserves $0.9$ of each audio-block interval for the three application copies, before the separate beacon reserve. Adding receiver reports, positioning exchanges, or longer slot margins reduces scheduling slack.

These observations define a conditional decision sequence: compare modeled actuator resolution with the required electronic timing; determine whether geometry and information age fit the remaining acoustic allowance; evaluate fixed-copy group risk for the specified channel parameters; and test reserved slot demand against the schedule budget and playback horizon. Satisfaction of only the first or third condition is insufficient. The numerical results quantify separate timing, reception, and reservation constraints rather than a jointly sampled closed-loop success probability.

\subsection{Failure regimes: common fades and overload}
Fig.~\ref{fig:failure} preserves unfavorable regimes. At $N=16$, raising bad-state common erasure probability from 0.2 to 0.6 reduces fixed-three worst-receiver delivery from about 0.99783 to 0.98576; worst-risk falls from about 0.99133 to 0.97767. The nominal 0.01 risk proxy is not a guaranteed constraint. Clipping repetitions at three, updating only every 0.2~s, and extrapolating an IID marginal all contribute to the failure. This illustrates the sensitivity of deadline-dependent adaptation to hidden channel state and confidence in loss estimates, rather than a guaranteed benefit from increasing control complexity.

Capacity imposes an independent limitation. At the main 150~$\mu$s slot assumption, fixed-three load is 0.91 including the beacon reserve. At 200~$\mu$s it is 1.21, and at 250~$\mu$s it is 1.51. When reserved-slot length alone changes, the model reuses the same stochastic loss mechanism, so the radio-block success results do not change; those outcomes must not be interpreted as feasible audio delivery above unit load. Even a favorable mean load can hide short reporting intervals or high-copy epochs that violate instantaneous scheduling. The complete outputs retain maximum run load and overload fraction.

These failure regimes clarify the difference between a mathematically schedulable configuration and an overloaded one. The assumed slot margin is a design input rather than a guaranteed service capability, and correlated loss cannot be replaced by an optimistic IID packet count. When a fixed-copy policy satisfies an assumed reliability and deadline budget with smaller feedback overhead, the adaptive rule has no demonstrated advantage under that budget. The results support resource-aware design choices and do not imply categorical superiority of UWB over LE Audio.

\begin{figure*}[t]
\centering\includegraphics[width=\textwidth]{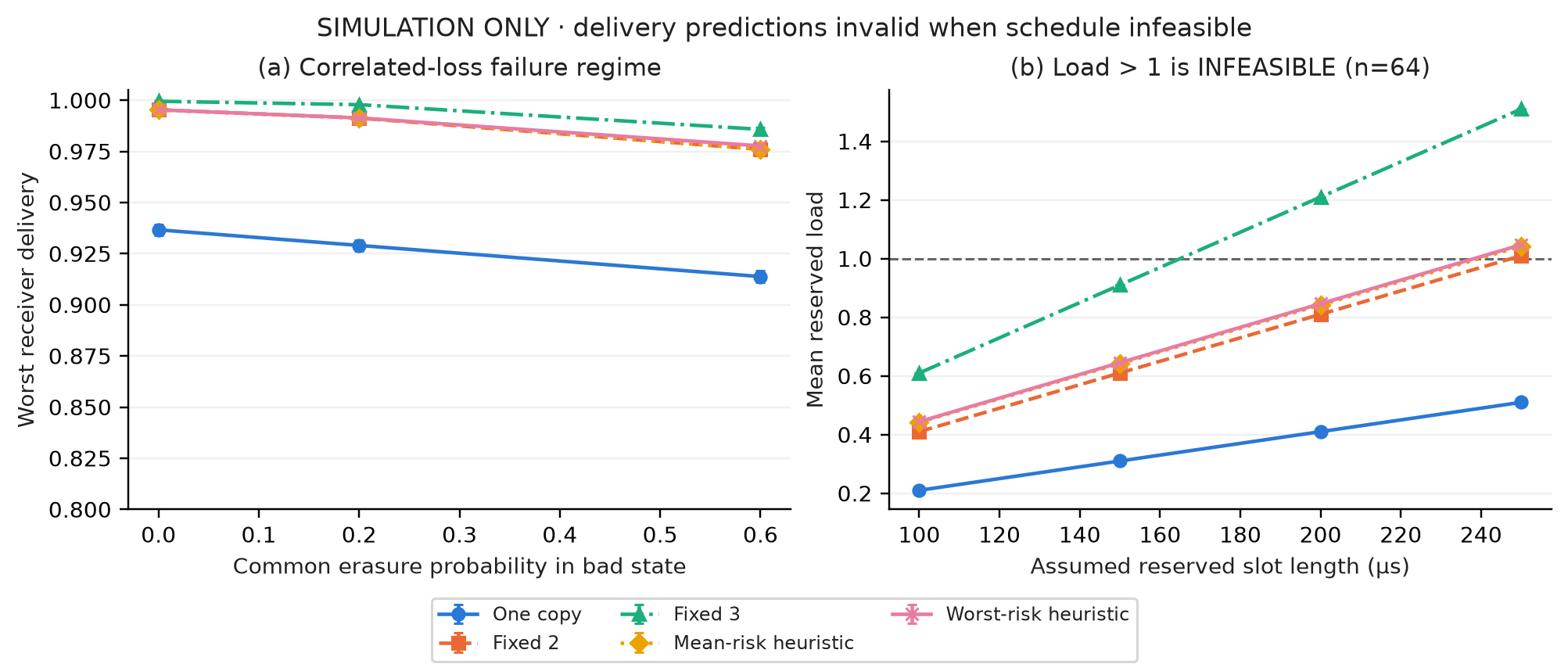}
\caption{Simulation-only failure regimes. Left: stronger shared fades expose the lack of a correlated-risk guarantee. Right: larger assumed reserved slots create offered load above one; such points are infeasible and do not predict real-time delivered audio. Thirty independent runs per setting.}\label{fig:failure}
\end{figure*}

\section{Discussion and Design Implications}\label{sec:discussion_main}
\subsection{A Common Feasibility Region}
The central implication of the analysis is that multicast-audio performance is determined by the intersection of several constraints, not by a single wireless metric. The clock estimator controls the temporal relation between source and receiver clocks; the output stage determines whether the estimated deadline becomes a specified sample event; the position estimator controls the speaker-dependent propagation offset; and the transport schedule determines whether the correct block is available early enough to render. These variables interact through deadlines but are observed at different layers. A low packet-error ratio therefore does not guarantee acoustic alignment, just as a low clock residual does not guarantee that an audio block was received in time. Equations~\eqref{eq:bound}, \eqref{eq:causal}, and~\eqref{eq:load} give distinct checks on a single operating point.

A timing bound and a delivery bound have different logical meanings. Proposition~\ref{prop:budget} is a sufficient statement conditional on bounded timing and geometry errors. It can certify that a specified spread is not exceeded within those premises, although failing the bound does not establish an actual timing failure because the bound can be conservative. In contrast, the common-erasure lower bound in Proposition~\ref{prop:groupbounds} is a necessary test: a target that lies below the shared-loss floor is impossible under the channel model regardless of the quality of the local receiver channels. The union bound is sufficient for group admission but may be loose when multiple receivers fail in the same common fade. Comparing these logical directions prevents a necessary test from being presented as a guarantee or a sufficient test from being interpreted as an exact prediction.

The resulting joint region is particularly informative at high receiver counts. Fixed-copy multicast uses a common source-data service for all receivers, so its downlink reservation does not grow linearly with $N$. However, the event that every receiver obtains a particular block becomes more demanding as receivers are added. The number of reports required by a feedback strategy also increases with $N$ under the assumed per-receiver return schedule. Hence receiver-independent downlink service, group reliability, and feedback scalability are separate properties. The multicast advantage in repeated content distribution does not eliminate the need to budget the reliability and control terms.

\subsection{Timing and Spatial Interpretation}
For the specified one-listener objective, source-time emission is deliberately shifted to offset unequal distances. The correction changes the variable of interest: electrical outputs need not coincide after compensation, but direct-path arrivals can approach one another at the chosen listener. This distinction is essential when interpreting the timing tables. An uncorrected electrical span and a corrected acoustic span quantify different objectives and cannot be compared as if both were residual clock jitter. Likewise, a sub-sample numerical residual does not imply that a sample-addressable physical output stage has the same resolution; quantization and output-pipeline uncertainty must be represented in the timing budget.

The bound also exposes the marginal value of additional sensing accuracy. Suppose a target arrival spread is $S^*$ and the actuator contributes $q_D$. After assigning $2E_c$ to clock uncertainty, the remaining allowable one-way distance error under the symmetric deterministic budget is $c(S^*-2E_c-q_D)/2$. This allowance is meaningful only if nonnegative. If clock and actuation already consume the entire bound, more precise position coordinates cannot establish the desired guarantee. Conversely, once the clock contribution is small, position-age error or actuation granularity may dominate. These are structural properties of the propagation and output mappings, independent of how the coordinates were obtained.

Sensor age is not identical to a fixed application update interval. The effective age at playback reflects when the source position was observed relative to the output event, including any estimation and processing delay. Under the bounded-speed model, age $a$ introduces a worst-case displacement contribution $va$ even if the underlying coordinate observation were exact. An age-reduction policy can therefore be more beneficial than reducing instantaneous coordinate noise in a rapidly changing geometry. A compensator designed for one listener does not in general minimize error for a second listener at another position; the corresponding path-length differences define distinct geometric objectives. Reverberation, reflections, transducer phase, and frequency-dependent directivity are also outside the direct-path arrival metric used here and are not implicitly summarized by it.

At 48~kHz, one sample occupies approximately $20.833\us$. If this full quantum is assigned to the actuator term and the study threshold is $30\us$, only $9.167\us$ remains in the sufficient worst-case budget for all clock and distance uncertainties. Under zero clock error, the symmetric per-path distance allocation is approximately $1.57$~mm. This is a conditional budget calculation rather than a localization specification: other choices of interpolation, actuator control, or probabilistic rather than deterministic criteria lead to other allowances. Its significance is that finer radio timing alone does not eliminate uncertainty introduced downstream of the radio.

\subsection{Correlated Loss, Repetition, and Group Size}
The shared-loss mechanism is decisive because one common erased copy simultaneously affects every receiver. Under a fixed channel state, repetitions can overcome local erasures, but state persistence produces bursts whose duration is not reflected in an unconditional single-copy loss probability. The stationary fixed-repetition expression conditions on the common loss process before combining receiver outcomes. Averaging the losses first and then raising the marginal loss to the repetition power implicitly imposes a different temporal dependence structure. This ordering can materially overstate reliability when the bad state persists over several copy opportunities.

The group metric also changes the interpretation of apparently excellent individual reception. If every receiver is required to render the same block, the service event is an intersection over receivers rather than an average over them. A Jain-style equality statistic can approach one when all receivers suffer the same missing blocks, while the useful group service remains poor. In systems that tolerate partial attendance or receiver-specific content, a different service metric is appropriate; in synchronized common-block playback, simultaneous delivery is the relevant quantity. The analysis separates this requirement from fairness, perceived audio quality, and buffer-concealment behavior, none of which follows from the block-delivery event alone.

Repetition policy should be interpreted jointly with reserved load and deadline slack. Choosing a larger copy count increases packet opportunities in the loss model but reserves additional slots whether or not earlier copies already succeeded. A per-receiver report can support adaptation, yet its reserved return-channel cost rises with population size. A mean-risk rule inferred from marginal losses is especially vulnerable to state-correlated failures: the estimator may select a count that looks adequate under an independent-erasure proxy while the common-outage floor remains above the target. The numerical comparison therefore favors neither adaptability nor fixed repetition as a general principle; it shows which configured policies meet or miss explicit conditional objectives at their respective resource costs.

\subsection{Schedule Capacity and Causal Playback}
A reserved-slot utilization exceeding one is a feasibility contradiction within the specified service calendar. It is not simply a performance point at which a queue would experience a small additional delay. If all stated reservations must be honored in the same interval and their combined duration exceeds the interval, the assumed nonoverlapping schedule cannot exist. For utilization below one, feasibility of the aggregate reservation is necessary but not sufficient for deadline satisfaction: placement of repetitions, timing beacons, reporting opportunities, decoding, and output preparation still matters. This is why a sufficient delivery probability for a block is not interchangeable with successful presentation at the block's scheduled acoustic arrival time.

The causal condition~\eqref{eq:causal} formalizes the distinction between receiving a block eventually and receiving it early enough to render. Extra copies near the end of the playback window can improve an unconstrained reception statistic but provide no useful service for a missed presentation instant. Similarly, transmitting a timing or position update consumes time that could otherwise be used for audio service; its value is determined by whether the reduction in uncertainty outweighs the schedule and age penalties. A cross-layer design that ignores the execution horizon can optimize the wrong deadline. A coherent formulation retains the presentation event, not packet reception, as the terminal timing boundary.

\subsection{Interpretation and Scope of the Numerical Evidence}
All reported confidence intervals quantify variation over independent Monte Carlo runs at the specified parameter values. The repeated policies within a seed use common candidate events and should be treated as paired comparisons rather than unrelated experiments. The clock and transport model families are evaluated separately; their output statistics are marginals under distinct stochastic processes. In the absence of a specified joint distribution of delay, erasure, resampling, and playout-buffer events, multiplying those marginals would not yield a defensible end-to-end audio success probability. The study instead reports the separately identifiable operating limits and their analytical intersection.

The geometry model includes one designated listening point, known speaker positions, bounded listener speed, fixed acoustic propagation speed, and stochastic external position estimates. The channel model uses a two-state common erasure process and conditionally local losses. These assumptions permit controlled sensitivity studies and exact conditional expressions, but the confidence intervals do not measure uncertainty about model structure or the distribution of different rooms and radios. Likewise, calculated radio airtime and reserved receiver activity are distinct from energy consumption; their interpretation follows the explicitly defined accounting boundary. Published Bluetooth isochronous behavior provides context, not a matched technology ranking, because the numerical policies do not implement a standards-complete comparison stack.

The main transferable result is thus a set of checks for service claims. A configuration should be characterized by the precision and identifiability of its time observations, its output-phase actuation resolution, its spatial and information-age assumptions, the group event used to define reliability, and its service-calendar utilization. None of these terms can be replaced by another layer's headline metric. For multi-speaker audio, this accounting makes it possible to state which error mechanism limits a declared objective, which reliability targets are excluded by common outages, and when the apparent benefit of added redundancy is incompatible with the available schedule.

\section{Conclusion}\label{sec:conclusion}
This paper formulated position-assisted UWB multicast audio as a joint timing, group-reliability, and schedule-feasibility problem. The analysis separated clock mapping and output actuation from geometric propagation correction, and derived a conditional one-listener arrival-span budget. For state-correlated common erasures and heterogeneous receiver-local losses, the stationary all-receiver delivery expression and complementary probability bounds distinguished reliable group service from receiver-average packet reception. A causal output constraint and reserved-slot load expression made explicit the resource and deadline implications of repeated broadcast transmission.

Across reproducible numerical experiments, affine clock compensation reduced electronic spread but did not by itself establish the stipulated acoustic alignment target. Position-error age and output granularity were decisive for the one-listener correction, while common fades, receiver population, and repetition budgets limited simultaneous block delivery. The resulting design principle is to assess a shared audio transport against declared presentation-time, all-receiver delivery, and resource constraints together. This characterization provides a coherent basis for interpreting multicast audio architectures without conflating radio packet reception, scheduled electrical output, and acoustic arrival.

\bibliographystyle{IEEEtran}
\bibliography{references/references}
\end{document}